\documentclass[%
 reprint, superscriptaddress,
 nofootinbib,
 amsmath, amssymb,
 aps, pra, longbibliography,
 floatfix, twocolumn%
]{revtex4-2}

\usepackage[linesnumbered,boxed]{algorithm2e}
\usepackage{graphicx}
\usepackage{dcolumn}
\usepackage{bm}
\usepackage{multirow}
\usepackage{diagbox}
\usepackage[dvipsnames]{xcolor}
\usepackage{hyperref}
\usepackage{tcolorbox}
\usepackage{enumerate}
\usepackage{enumitem}
\setlist[enumerate,1]{label=\arabic*.,font=\textup,
leftmargin=4mm}

\hypersetup{%
colorlinks = true,
urlcolor  = BrickRed,
linkcolor = NavyBlue,
citecolor = ForestGreen%
}

\def\ket#1{\left| #1 \right\rangle}
\def\bra#1{\left\langle #1 \right|}

\newcommand{\beq}{\begin{equation}}
\newcommand{\eeq}{\end{equation}}

\usepackage{amsthm}
\newtheorem{theorem}{Theorem}
\newtheorem{proposition}{Proposition}
\newtheorem{lemma}[theorem]{Lemma}
\newcommand{\Tr}{\operatorname{Tr}}

\begin{document}
\title{Photonic verification of device-independent quantum key distribution against collective attacks}

\author{Wen-Zhao Liu}
\author{Yu-Zhe Zhang}
\author{Yi-Zheng Zhen}
\author{Ming-Han Li}
\affiliation{Hefei National Laboratory for Physical Sciences at Microscale and Department of Modern Physics, University of Science and Technology of China, Hefei 230026, P.~R.~China.}
\affiliation{Shanghai Branch, CAS Center for Excellence and Synergetic Innovation Center in Quantum Information and Quantum Physics, University of Science and Technology of China, Shanghai 201315, P.~R.~China.}
\affiliation{Shanghai Research Center for Quantum Sciences, Shanghai 201315, P.~R.~China.}
\author{Yang Liu}
\affiliation{Jinan Institute of Quantum Technology, Jinan 250101, P.~R.~China}

\author{Jingyun Fan}
\affiliation{Shenzhen Institute for Quantum Science and Engineering and Department of Physics, Southern University of Science and Technology, Shenzhen, 518055, P.~R.~China}

\author{Feihu Xu}
\author{Qiang Zhang}
\author{Jian-Wei Pan}
\affiliation{Hefei National Laboratory for Physical Sciences at Microscale and Department of Modern Physics, University of Science and Technology of China, Hefei 230026, P.~R.~China.}
\affiliation{Shanghai Branch, CAS Center for Excellence and Synergetic Innovation Center in Quantum Information and Quantum Physics, University of Science and Technology of China, Shanghai 201315, P.~R.~China.}
\affiliation{Shanghai Research Center for Quantum Sciences, Shanghai 201315, P.~R.~China.}

\begin{abstract}
The security of quantum key distribution (QKD) usually relies on that the users's devices are well characterized according to the security models made in the security proofs. In contrast, device-independent QKD --- an entanglement-based protocol --- permits the security even without any knowledge of the underlying quantum devices. Despite its beauty in theory, device-independent QKD is elusive to realize with current technologies. Especially in photonic implementations, the requirements for detection efficiency are far beyond the performance of any reported device-independent experiments. In this paper, we employ theoretical and experimental efforts and realize a proof-of-principle verification of device-independent QKD based on the photonic setup. On the theoretical side, we enhance the loss tolerance for real device imperfections by combining different approaches, namely, random post-selection, noisy preprocessing, and developed numerical methods to estimate the key rate via the von Neumann entropy. On the experimental side, we develop a high-quality polarization-entangled photon source achieving a state-of-the-art (heralded) detection efficiency about $87.5\%$. This efficiency outperforms previous photonic experiments involving loophole-free Bell tests. Together, we show that the measured quantum correlations are strong enough to ensure a positive key rate under the fiber length up to 220~m. Our photonic platform can generate entangled photons at a high rate and in the telecom wavelength, which is desirable for high-speed generation over long distances. The results present an important step towards a full demonstration of photonic device-independent QKD.
\end{abstract}
\maketitle

\emph{Introduction. ---}
Quantum key distribution (QKD)~\cite{bennett1984quantum,ekert1991quantum} allows two distant users to share secret keys with information-theoretic security~\cite{xu2020secure}. The security of QKD usually relies on the assumptions that the devices are trusted and well-characterized~\cite{lo1999unconditional,shor2000simple,renner2008security}. However, the imperfections of the practical devices may provide potential backdoors or side channels for adversaries~\cite{xu2010,Lars2010}. Measurement-device-independent QKD~\cite{lo2012measurement,braunstein2012side} (with a recent development in~\cite{lucamarini2018overcoming}) was proposed to prevent the side-channel attacks on detectors, but leaves the state-preparation devices to be precisely calibrated. Device-independent QKD~\cite{mayers1998quantum,barrett2005no,acin2007device,pironio2009device} further relaxes the security assumptions on the devices. Given the following assumptions~\cite{pironio2009device}, i.e., (i) quantum theory is validity, (ii) no unwanted information leakage from communicating parties to adversaries is allowed, (iii) the communicating parties have local trusted randomness to decide inputs of their measurement devices, (iv) the classical post-processing units are trusted, and (v) an authenticated public classical channel is shared between the communicating parties, its security can be guaranteed solely based on the violation of Bell inequalities.

The realization of device-independent QKD is non-trivial with current technologies, where the loophole-free violations of the Bell inequalities are usually required~\cite{hensen2015loophole,rosenfeld2017event}. Especially in the photonic implementations, the efficiency loss of photons due to transmission and detection becomes a key issue. Although distinguished experiments without the detection loophole have been made ~\cite{giustina2013bell,christensen2013detection,shalm2015strong,giustina2015significant,li2018test,liu2018high,shen2018randomness,bierhorst2018experimentally,liu2018device,zhang2020experimental,liu2021device,li2021experimental,shalm2021device}, a much higher efficiency (over $90\%$) is normally required in the realization of device-independent QKD ~\cite{pironio2009device,masanes2011secure,reichardt2013classical,vazirani2014fully,Arnon2018Practical}. Despite recent theoretical progresses have been made in reducing the required efficiency~\cite{fourvalue,tan2020advantage,ho2020noisy,woodhead2021device,sekatski2021device,gonzales2021device,brown2021computing,schwonnek2021device,computeDI2021}, a practical protocol for a real platform remains elusive.

Here, we report a proof-of-principle verification of device-independent QKD based on polarization-entangled photons. We accomplish this via significant theoretical and experimental efforts. On the theoretical side, we propose a protocol that greatly enhances the loss tolerance of the experiments, thereby reducing the efficiency threshold of our setup to about $86\%$. The idea of our protocol is to post-select the outcomes of the key generation basis~\cite{de2016randomness,xu2021device}, and then add noise~\cite{ho2020noisy} to the remaining strings. The lower bound of the key rate is computed via the recent achievement in estimating the quantum conditional entropy~\cite{computeDI2021}. On the experimental side, we develop an entangled-photon source with the state-of-the-art efficiency of about $87.5\%$, which surpasses the values reported in previous full-photonic experiments that performing loophole-free Bell tests ~\cite{giustina2013bell,christensen2013detection,shalm2015strong,giustina2015significant,li2018test,liu2018high,shen2018randomness,bierhorst2018experimentally,liu2018device,zhang2020experimental,liu2021device,li2021experimental,shalm2021device}, and makes the device-independent experiments possible.
Combining the experimental and theoretical advances, we present a verification of device-independent QKD under fiber length up to 220~m. To maintain a high efficiency, our proof-of-principle experiment does not include the real-time random basis selection~\cite{liu2021device,li2021experimental,shalm2021device} and the finite-key effect, which are essential to generate real secret keys. However, our experiment verifies that the measured correlations are strong enough to guarantee a positive secret key rate, thus presenting an important step towards a full demonstration.

\begin{tcolorbox}[title = \textbf{Box 1 : The device-independent QKD protocol with random post-selection and noisy preprocessing.},colback=white, colframe=black!50!white]

\textbf{Assumptions:~\cite{pironio2009device}}\\
In addition to the assumptions under the device-independent QKD regime mentioned above, the devices are memoryless and behave identically and independently (i.i.d.) during the implementation.

\begin{flushleft}
 \textbf{Arguments:}\\
 $N$ --  number of rounds \\
 $p$ -- post-selection probability to keep a bit ``1" \\
 $p_N$ -- noisy preprocessing probability to flip a bit
\end{flushleft}
\tcblower

\textbf{Protocol:}
\vspace{-0.2cm}
\begin{enumerate}
\item For every round $i$, Alice and Bob agree that part of rounds corresponding to $(\bar{x}_i,\bar{y}_i)=(1,3)$ are the ``key-generation round" to generate the string of raw keys, and the rest corresponding to $(x_i,y_i)\in\{1,2\}\times\{1,2,3\}$ are the ``test round" to test the nonlocal correlations. The rest steps are all performed on the rounds corresponding to $(\bar{x}_i,\bar{y}_i)$.
\item {\it Random post-selection.}~\cite{xu2021device} Alice and Bob each randomly and independently discard bits ``1" with probability $1-p$, and keep all the rest bits (containing all bits ``0" and $p$ of bits ``1").
\item Alice and Bob announce the discarded rounds via an authenticated public channel, and keep rounds not mentioned by either party.
\item {\it Noisy preprocessing.}~\cite{ho2020noisy} Alice generates the noisy raw keys $\hat{a}_{\bar{x}}$ by flipping each of her remaining bits with probability $p_{N}$, where $\hat{a}$ denotes a bit after noisy preprocessing and subscript $\bar{x}$ represents that it is post-selected.
\item {\it Error correction and Privacy amplification.} After a one-way error correction protocol and a privacy amplification procedure, the secret keys can be distilled.
\end{enumerate}
\end{tcolorbox}

\emph{Protocol. ---} Our protocol is a modification of the protocol described in~\cite{pironio2009device}. As shown in Fig.~\ref{fig1}, entangled photon pairs are shared between Alice and Bob. Considering each of the $N$ rounds of experiments shown in Fig.~\ref{fig1}, Alice randomly choose binary input $x\in\{1,2\}$ and obtain binary outcome $a\in\{0,1\}$, and Bob randomly choose triple input $y\in\{1,2,3\}$ and obtain binary outcome $b\in\{0,1\}$, where $a,b = 0$ denotes a ``click" event on the respective detector, and $a,b = 1$ denotes a ``no-click" event. The total probabilities of joint measurement for outcomes $(a,b)$ and inputs $(x,y)$ are denoted as $P(a,b|x,y)$.


Given the raw outcomes $P(a,b|x,y)$, we further introduce the random post-selection and noisy-preprocessing approaches before distilling the final keys (see section A.1 of Appendix for details). We randomly set part of the $N$ rounds corresponding to the measurement inputs $(\bar{x},\bar{y})=(1,3)$ as ``key-generation round" and the rest as ``test round" to test nonlocal correlations, where $\bar{x},\bar{y}$ represent the input for ``key-generation round". First, we conduct the random post-selection $\mathcal{V}$ over all the ``key-generation round", where Alice and Bob each randomly and independently discard bits ``1" with probability $1-p$, and keep all the rest bits (containing all bits ``0" and $p$ of bits ``1"). Then, both Alice and Bob announce the discarded rounds via an authenticated public channel. Those renounced ``key-generation round" from either Alice or Bob will be simultaneously discarded by both parties, regardless of the outcome of the other side. Next, Alice further locally performs a noisy preprocessing $\mathcal{N}$, to generate the noisy raw keys by flipping each bit of her remaining string with probability $p_N$. Finally, an error correction step allows Alice to share the raw key with Bob, and secret keys can be distilled after privacy amplification. The full procedures of our protocol is listed in Box~1.

\begin{figure}[htbp]
    \centering
    \includegraphics[width=1.0\linewidth]{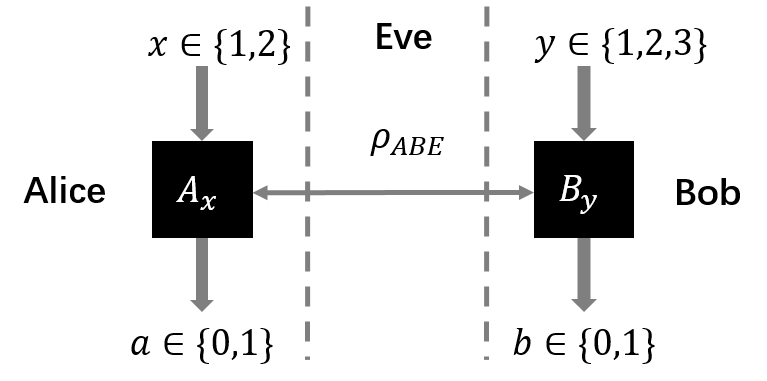}
    \caption{\label{fig1} An illustration of the device-independent QKD protocol. Alice and Bob share a pair of entangled photons potentially controlled by Eve ($\rho_{ABE}$). Alice performs a measurement to her share with binary input $x\in\{1,2\}$ and binary output $a\in\{0,1\}$. Bob performs a measurement to his share with triple input $y\in\{1,2,3\}$ and binary output $b\in\{0,1\}$.
    }
\end{figure}

We remark that the random post-selection effectively remove a fraction of ``non-click" events that contain few correlations and lots of errors, and therefore suppress the cost of error correction~\cite{xu2021device}. The noisy preprocessing decrease the correlations between Alice and Eve by mixing the probability distributions with randomness~\cite{ho2020noisy}. These two approaches jointly contribute to the enhancement of experimental loss tolerance (see section A.1 of Appendix for details).

\emph{Key rate estimation. ---} We consider the collective attack model where the devices behave in an independent and identically distributed (i.i.d.) manner and the devices are memoryless~\cite{barrett2013memory,curty2017quantum} during the implementation of the protocol. In the process of random post-selection, given the outcomes $(a,b)$ and the definition that $p_{\alpha}=1\cdot\Delta_{\alpha,0}+p\cdot\Delta_{\alpha,1}$ for a certain ``key-generation round", where $\Delta_{\alpha,i}$ represent the Dirac delta function in a form of $\delta(\alpha-i)$, the probability it can be retained is given by $p_{\mathcal{V}_p}=\sum_{(a,b)\in\mathcal{V}_p}{\omega_{ab}P(a,b|\bar{x},\bar{y})}$, where $\mathcal{V}_p$ represents the set of post-selected rounds and $\omega_{ab} = p_{a}\cdot p_{b}$. In the infinite-key scenario, given the set of bipartite correlations $\{P(a,b|x,y)\}$ that character the devices, the secret key rate $r$ with error correction can be lower-bounded by the Devetak-Winter rate \cite{DW2005},
\begin{equation}
\label{keyrate1}
    r\ge p_{\mathcal{V}_p}\left[H(\hat{A}_{\bar{x}}|E,\mathcal{V}_p\rightarrow\mathcal{N}_{p_N})-f_{e}H(\hat{A}_{\bar{x}}|B_{\bar{y}},\mathcal{V}_p\rightarrow\mathcal{N}_{p_N})\right],
\end{equation}
where $\mathcal{N}_{p_N}$ denotes the set of string after noisy preprocessing given post-selected set $\mathcal{V}_p$, $\hat{A}_{\bar{x}}$ denotes the noisy raw key of Alice after random post-selection and noisy preprocessing, $H(\hat{A}_{\bar{x}}|E,\mathcal{V}_p\rightarrow\mathcal{N}_{p_N})$ represents the single-round conditional von Neumann entropy that quantifies the strength of the correlations between Alice and Eve, $H(\hat{A}_{\bar{x}}|B_{\bar{y}},\mathcal{V}_p\rightarrow\mathcal{N}_{p_N})$ represents the single-round cost of one-way error correction from Alice to Bob, and $f_{e}$ is the error correction efficiency.

As a proof-of-principle verification, we consider the perfect error correction with Shannon limit $f_e = 1.0$, which is reachable in the case of infinite data size (see section A.3 of Appendix for details). We then adopt the method in Ref.~\cite{computeDI2021} to show that the single-round conditional von Neumann entropy $H(\hat{A}_{\bar{x}}|E,\mathcal{V}_p\rightarrow\mathcal{N}_{p_N})$ can be bounded by a converging sequence of optimizations that can be subsequently computed using the NPA hierarchy~\cite{NPA} (see section A.2 of Appendix for details). Note that for all ``test round", Alice and Bob save the outcomes without any post-selections since the Bell tests are implemented without detection loophole~\cite{de2016randomness}. (For more detailed security proof of protocol, please refer to section A.3 of Appendix and Ref.~\cite{xu2021device}.)

\begin{figure*}[thbp]
\centering
\resizebox{15cm}{!}{\includegraphics{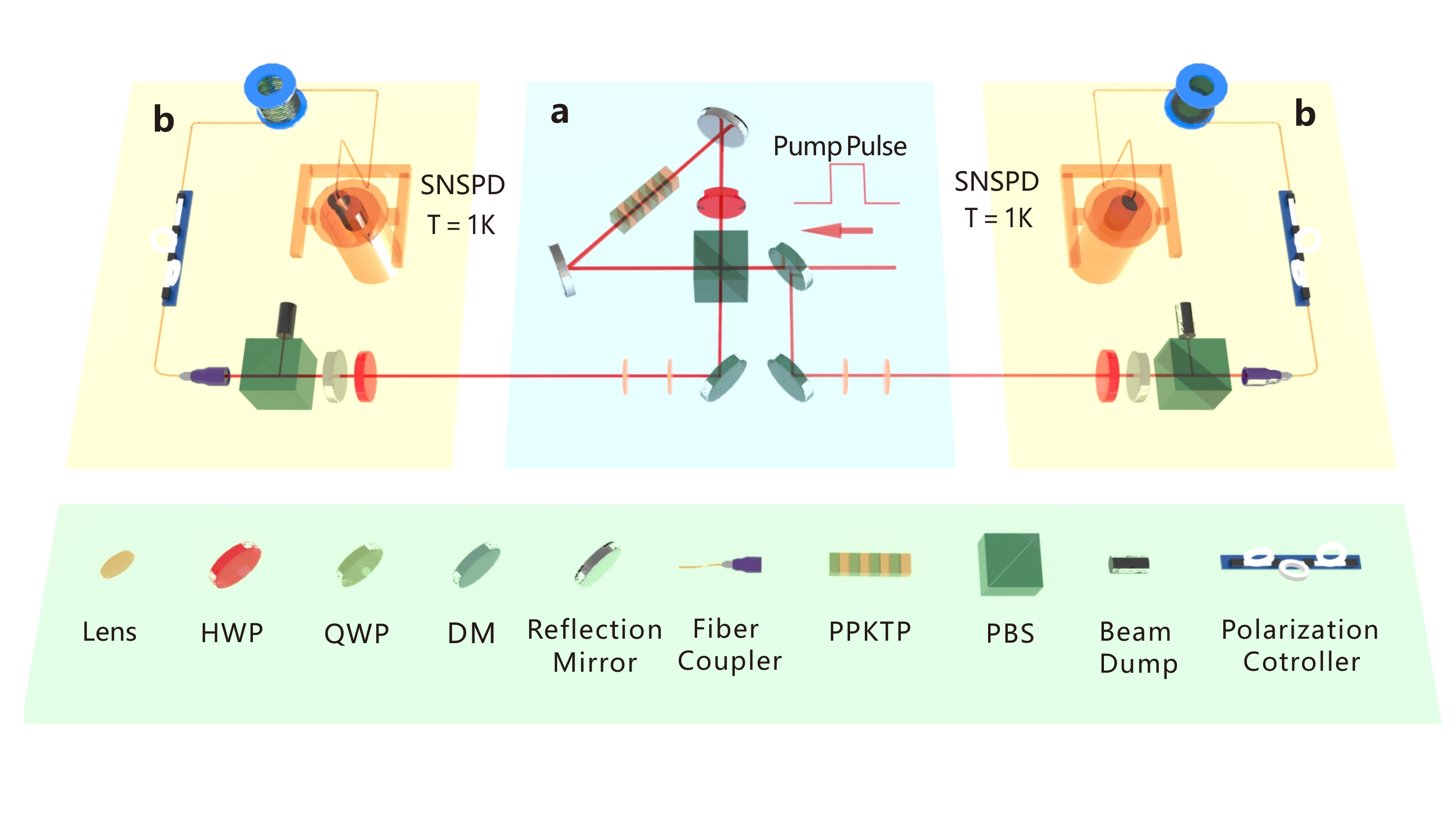}}\vspace{-0.5cm}
\caption{{\bf Schematic of the experiment.} {\bf a} Entanglement Source, Creation of pairs of entangled photons: Light pulses of 10~ns are injected at a repetition pulse rate of 2~MHz into a periodically poled potassium titanyl phosphate (PPKTP) crystal in a Sagnac loop to generate polarization-entangled photon pairs. The two photons of an entangled pair at 1560~nm travel in opposite directions to two sites Alice and Bob, where they are subject to polarization projection measurements. The PPKTP is placed in the middle of the hypotenuse of the Sagnac loop with a small angle to the light path, which does not significantly affect the upper limit of efficiency that the photonic setup could achieve, but can suppress the reflection of imperfect devices for 1560~nm photons. These enhancements lead that the non-maximally entangled state generated in our experiment has a better fidelity $99.52 \pm 0.15\%$ as compared to our previous work~\cite{liu2018device,li2021experimental,liu2021device}. {\bf b} Alice and Bob, single-photon polarization measurement: In the measurement sites, Alice (Bob) uses a set of HWP and QWP to project the single photon into pre-determined measurement bases. After being collected into the fiber, the single photons transmit through a certain length of fiber and then are detected by a superconducting nanowire single-photon detector (SNSPD) operating at 1K. HWP -- half-wave plate; QWP -- quarter-wave plate; DM -- dichroic mirror; PBS -- polarizing beam splitter.
}
\label{fig:setup}
\end{figure*}

\begin{table}[htbp]
\caption{{Efficiencies in existing photonic experiments of {\em loophole-free} Bell tests and related applications. The efficiencies in the table are averaged over Alice's and Bob's global detection efficiency. (QRNG: quantum random number generation)} }
\begin{tabular}{c|ccc|c}
\hline \hline
Label & Experiment                                          & Year & Type                       & Efficiency \\ \hline
(1)   & Shalm {\em et al.}~\cite{shalm2015strong}                 & 2015 & Bell test            & $75.15\%$  \\
(2)   & Giustina {\em et al.}~\cite{giustina2015significant}      & 2015 & Bell test            & $77.40\%$   \\
(3)   & Liu {\em et al.}~\cite{liu2018high}                       & 2018 & QRNG                 & $79.40\%$   \\
(4)   & Shen {\em et al.}~\cite{shen2018randomness}               & 2018 & QRNG                 & $82.33\%$  \\
(5)   & Bierhorst {\em et al.}~\cite{bierhorst2018experimentally} & 2018 & QRNG                 & $75.50\%$   \\
(6)   & Liu {\em et al.}~\cite{liu2018device}                     & 2018 & QRNG                 & $78.65\%$  \\
(7)   & Li {\em et al.}~\cite{li2018test}                         & 2018 & Bell test            & $78.75\%$  \\
(8)   & Zhang {\em et al.}~\cite{zhang2020experimental}           & 2020 & QRNG                 & $76.00\%$   \\
(9)   & Shalm {\em et al.}~\cite{shalm2021device}                 & 2021 & QRNG                 & $76.30\%$   \\
(10)   & Li {\em et al.}~\cite{li2021experimental}                & 2021 & QRNG                 & $81.35\%$   \\
(11)   & Liu {\em et al.}~\cite{liu2021device}                    & 2021 & QRNG                 & $84.10\%$   \\
(12)   & This work                                                & 2021 & QKD                  & $87.49\%$  \\
\hline \hline
\end{tabular}
\label{tab:eff}
\end{table}

\emph{Experiment. ---}
A schematic of the experiment is depicted in Fig.~\ref{fig:setup} which consists of three modules. Pairs of polarization-entangled photons at the wavelength of 1560 nm are generated probabilistically via the spontaneous parametric downconversion process in the central module (a). The pairs of photons are sent to two side modules (b), where Alice and Bob perform correlated detections to generate secret keys. The single-photon detection efficiency is respectively determined to be $87.16\pm0.22\%$ and $87.82\pm0.21\%$ for Alice and Bob (see section B.1 of Appendix for details), which significantly surpass the record values in previous loophole-free Bell tests with photons~\cite{shalm2015strong,giustina2015significant,li2018test,liu2018high,shen2018randomness,bierhorst2018experimentally,liu2018device,zhang2020experimental,liu2021device,li2021experimental,shalm2021device} (see Table~\ref{tab:eff}). Furthermore, the values also surpass the efficiency threshold of $86.2\%$ for device-independent key generation in a realistic scenario (see section A.3 of Appendix for details).

According to the numerical studies, we prepare a non-maximally two-photon entangled state $\cos(20.0^\circ)\ket{HV}+\sin(20.0^\circ)\ket{VH}$ and set the measurement settings to $\{-88.22^\circ, 54.29^\circ$\} and $\{9.75^\circ, 21.45^\circ, -1.07^\circ\}$ respectively for $x\in\{1,2\}$ and $y\in\{1,2,3\}$ to optimize the probability of key generation, where the values presented in degree are angles of half-wave plates in the polarization measurements by Alice and Bob (see Fig.~\ref{fig:setup}). We experimentally measure a two-photon state fidelity of $99.52\pm0.15\%$ with respect to the ideal state and achieve a CHSH game winning probability of $0.7559$ with optimized state and measurement settings (see section C.1 of Appendix for details), both substantially improving over previous results~\cite{liu2018high, liu2018device,li2018test,li2021experimental,liu2021device}. We repeat the experiment at a rate of $2\times10^6$ rounds per second.

Note that as the device-independent QKD itself assumes the validity of quantum theory and that Alice and Bob have trusted random number generators. Nevertheless, the device-independent QKD requires that, in the entire duration of the protocol, the information about the input choices and output results of one party must be unknown at other locations, e.g., the Eve’s location. The parties must be well isolated to forbid any unwanted information leakage~\cite{pironio2009device}, which ensures that the locality loophole is closed. In our experiment, this is done via the shielding assumption~\cite{pironio2010random, liu2021device}, which prohibits unnecessary communications between untrusted devices and a potential adversary. For a more definitive experiment to eliminate the shielding assumption, one could use developed electromagnetic shielding techniques, such as materials including sheet metal, metal screen, and metal foam, to avoid possible unwanted information leakage. However, considering the essential photonic channels from entanglement source to both parties, perfect shielding might not be realized experimentally. To reduce experimental complexity, we also alternate the measurement settings instead of randomization to reduce experimental complexity. While these simplifications can not be adopted in a real-field application of device independent QKD, as we will show, our results demonstrate that the secure key generation is almost achievable using the state-of-the-art technologies.

We conduct $2.4\times10^8$ rounds of experiment for each of the six combinations of measurement settings ($x,y$) and perform data analysis following the protocol. With optimized parameters $p_N = 0.13$ and $p = 0.96$, we obtain $H(\hat{A}_{\bar{x}}|E,\mathcal{V}_p\rightarrow\mathcal{N}_{p_N}) = 0.56021$ and $H(\hat{A}_{\bar{x}}|B,\mathcal{V}_p\rightarrow\mathcal{N}_{p_N}) = 0.55995$. (see section C.2 of Appendix for details). Finally, according the asymptotic key rate in Eq.~\eqref{keyrate1}, $55,920$ bits of secret keys are expected to be distilled after error correction and privacy amplification. This corresponds to $2.33\times10^{-4}$ bit per round. Furthermore, we show the feasibility to successfully generate secret keys at a fiber length of 220 meters by conducting the same rounds of experiments, for which we re-optimize the experiment over $p_N$ and $p$. These results are shown in Tab.~\ref{tab:fiber}, where the drop of the key rate when increasing the fiber length is mainly due to the decreasing of overall efficiency.

\begin{table}[htb]
\caption{The secret key rate as a function of the fiber distance between Alice and Bob. We test the device-independent QKD protocol by adding different length of fibers. }
\begin{tabular}{c|c|c|c}
\hline \hline
Fiber length/m & Key rate/$bit{\cdot}{pulse}^{-1}$ & $p_N$  & $p$ \\ \hline
20             & $2.33\times10^{-4}$               & 0.13   & 0.96  \\ \hline
80             & $5.37\times10^{-5}$               & 0.17   & 0.94  \\ \hline
220            & $1.30\times10^{-6}$               & 0.49   & 0.99  \\ \hline \hline
\end{tabular}
\label{tab:fiber}
\end{table}

\emph{Conclusion. ---}
In conclusion, we demonstrate a proof-of-principle experiment of device-independent QKD against collective attacks using a full-photonic setup. The photonic implementation enjoys the advantages of high-rate entangled-photon generations in the telecom wavelength, which is important for the practical applications involving quantum memories or quantum repeaters forming a quantum internet. With a high-quality entangled photon source, we show the measured correlations are strong enough to guarantee a positive secret key rate. However, to actually produce a key, the real-time random basis selection and more information-processing processes, such as error correction, authentication and privacy amplification in the finite-key case, remains to be done.

For random basis selection, as implemented routinely by us and other groups~\cite{liu2021device,li2021experimental,shalm2021device}, it may normally introduce about $1\%$ additional efficiency loss, which indeed makes the system working at the marginal point of efficiency threshold. However, we remark that the performance of entangled system could be greatly improved via enhancing the fidelity of the entanglement state with a different type of design of the entanglement source~\cite{shalm2015strong}. This is possible to improve the fidelity from $99.5\%$ in our system to about $99.8\%$ (calculated by given visibility). With the improved fidelity, the required efficiency can drop to $84.8\%$. This would make it possible to introduce random basis switching.

We further remark that it is still tricky to realize a faithful photonic device-independent QKD with finite-key security. Apart from experimental technical difficulties, the protocol remains to be extended to the general-attack scenario~\cite{de2016randomness}. As we adopt three ingredients in the security analysis, i.e., random-post selection, noisy-preprocessing, and a numerical method to compute the lower bound of von Neumann entropy, the finite-key analysis involving all these ingredients need to be developed. However, the main problem is that, all experimental rounds might be correlated in a general-attack scenario, where Eve would learn more information of the remaining rounds from the discarded ones. This is similar to the problem encountered with the two-way communication protocol~\cite{tan2020advantage,Ma2006Twoway,Gottesman2003twoway}, where Alice and Bob have to randomly keep one of the selected two pairs of outcomes. Nonetheless, we noticed that there have been important theory developments in this direction~\cite{Bstep2021Ma}. It is foreseen that the finite-key analysis combining the method to compute von Neumann entropy with the random post-selection and noisy preprocessing will be significantly constructive in the future.


\emph{Note added.} When we are completing the manuscript, we notice two related works were completed based on trapped ions~\cite{ionDI2021} and trapped atoms~\cite{atomDI2021}.


\section*{Acknowledgments}
We thank Peter Brown, Roger Colbeck, Charles Lim, Nicolas Sangouard, Jean-Daniel Bancal, Yanbao Zhang, Xingjian Zhang for enlightening discussions. This work was supported by the National Key Research and Development (R\&D) Plan of China (2018YFB0504300,2020YFA0309701), the National Natural Science Foundation of China (617714438,62031024,12005091), the Anhui Initiative in Quantum Information Technologies, the Chinese Academy of Sciences and the Key-Area Research and Development Program of Guangdong Province (2020B0303010001, 2019ZT08X324).

Wen-Zhao Liu and Yu-Zhe Zhang contribute equally.

\onecolumngrid
\appendix
\section{Theoretical analysis}
\subsection{Models in device-independent QKD with random-post selection and noisy preprocessing}
Here we show the main results of our device-independent QKD protocol, which is a modified version based on Ref.~\cite{pironio2009device} that combines random post-selection~\cite{xu2021device} and noisy preprocessing~\cite{ho2020noisy} approaches. As shown in Fig.~\ref{fig1}, in each round, an entangled photon source emits entangled photon pairs to both measurement parties. Alice and Bob independently and randomly choose measurement settings $x\in\{1,2\}$ and $y\in\{1,2,3\}$ to obtain binary outcomes $a\in\{0,1\}$ and $b\in\{0,1\}$, respectively, where a click on the detector is denoted as bit ``0" and a non-click is denoted as bit ``1". According to the different inputs in each round, all rounds are divided into two parts, ``key-generation round" and `test round", where ``key-generation round" is part of the rounds with inputs $(\bar{x}=1,\bar{y}=3)$, and the rest are ``test round" with inputs $(x,y)\in\{1,2\}\times\{1,2,3\}$. All data preprocessing is performed on the set of ``key-generation round". The notation $(\bar{x},\bar{y})$ represents the inputs for ``key-generation round".
\begin{figure}[h]
\centering
\includegraphics[width=0.7\linewidth]{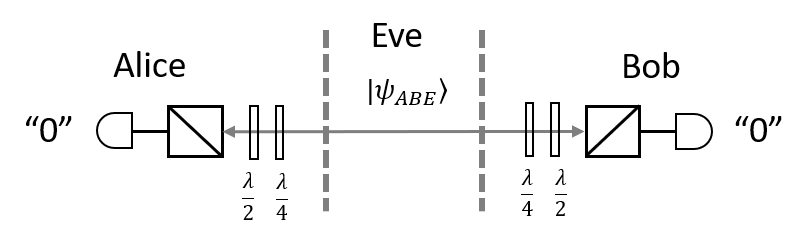}
\caption{\label{fig1}
{\bf An illustration of the device-independent QKD protocol.} The photon source $\ket\psi_{ABE}$ that is potentially controlled by adversary Eve, distribute entangled photon pairs to Alice and Bob. Alice and Bob use a set of half-wave plate $\lambda/2$ and quarter-wave plate $\lambda/4$ to choose the measurement settings $x,y$, and the binary output $a,b\in\{0,1\}$ can be obtained by the measurement of a polarized beam splitter and the following detector, respectively. $``0"$ denotes a click event and $``1"$ denotes a non-click one.}
\end{figure}

We derive the classical-quantum state after the random post-selection and noisy preprocessing and the cost of error correction. Since we focus on the asymptotic case under collective attacks with the assumptions that the devices are memoryless and behave identically and independently during the implementation, each round of the experiments can be analyzed individually. Considering the potential adversary, a tripartite state $\rho_{ABE} = \ket{\psi_{ABE}}\bra{\psi_{ABE}}$ is used to represent the quantum correlations among Alice, Bob, and Eve. The joint probability distribution corresponding to the bipartite measurement settings $(a,b)$ and outcomes $(a,b)$ can be described as
\begin{equation}
    P(a,b|x,y)=\Tr[(M^{A}_{a|x}\otimes M^{B}_{b|y} \otimes I_{E})\rho_{ABE}].
\end{equation}
where, $M^{A}_{a|x}$ and $M^{B}_{b|y}$ denote the positive operator-valued measures (POVMs) corresponding to Alice and Bob, respectively.

\emph{Random post-selection. ---}
In the implementation of random post-selection, Alice and Bob each randomly and independently discard bits ``1" with probability $1-p$, and keep all the rest bits (containing all bits ``0" and $p$ of bits ``1"). Therefore, a ``key-generation round" with outcomes $(a,b)$ can be retained with probability $\omega_{ab}=p_{a}\cdot p_{b}$, where $p_{\alpha}=1\cdot\Delta_{\alpha,0}+p\cdot\Delta_{\alpha,1}$ for $\alpha=a,b$ and $\Delta_{\alpha,i}$ is the Dirac delta function in the form of $\delta(\alpha-i)$.

Given the post-selected rounds $\mathcal{V}_p$, the three-party state $\rho_{ABE}$ can be expressed as
\beq
\rho_{ABE|\mathcal{V}_p}=\frac{1}{p_{\mathcal{V}_p}}\sum_{(a,b)\in {\mathcal{V}}_p}\omega_{ab}|ab\rangle\langle ab|\otimes \rho_{ab}^{E},
\eeq
where $\rho_{ab}^{E}=\Tr_{AB}[(M^{A}_{a|x}\otimes M^{B}_{b|y}\otimes I_{E})\rho_{ABE}]$ and $p_{\mathcal{V}_p}$ is the proportion of the post-selected events in a form of
\begin{align}
    p_{\mathcal{V}}=\sum_{a,b} \omega_{ab}P(a,b|\bar{x},\bar{y})=P(0,0|\bar{x},\bar{y})+p\cdot P(0,1|\bar{x},\bar{y})+p\cdot P(1,0|\bar{x},\bar{y})+p^2\cdot P(1,1|\bar{x},\bar{y}),
\end{align}
and the probability distribution $\widetilde{P}(a,b|\bar{x},\bar{y},\mathcal{V}_p)$ of the post-selected events $(a,b)\in\mathcal{V}_p$ is given by,
\begin{equation}
    \widetilde{P}(a,b|\bar{x},\bar{y},\mathcal{V}_p)=P(a,b|\bar{x},\bar{y})\cdot\omega_{ab}/p_{\mathcal{V}_p}.
\end{equation}
Meanwhile, by tracing out Bob's qubit, the reduced state $\rho_{AE|\mathcal{V}}$ is denoted as
\begin{equation}
    \rho_{AE|\mathcal{V}_p}=\frac{1}{p_{\mathcal{V}_p}}\left[|0\rangle\langle0|\otimes(\rho_{00}^{E}+p\cdot\rho_{01}^{E})+|1\rangle\langle1|\otimes(p\cdot\rho_{10}^{E}+{p^2}\cdot\rho_{11}^{E}) \right].
\end{equation}

We remark that before the random post-selection, Eve cannot know the exact remaining rounds and make quantum correlations due to the constraints of the device-independent QKD regime. After Alice and Bob announce the discarded rounds via the authenticated public channel, Eve would have to guess the remaining outcomes with limited information~\cite{de2016randomness}.

\emph{Noisy preprocessing. ---}
Following the noisy preprocessing $\mathcal{N}$, Alice generates the noisy raw keys $\hat{A}_{\bar{x}}$ by flipping each of her remaining bits independently with probability $p_{N}$, so that the reduced state becomes
\begin{align}
    \rho_{AE|\mathcal{V}_p\rightarrow\mathcal{N}_{p_N}}=&\frac{1}{p_{\mathcal{V}_p}}\{|0\rangle\langle0|\otimes\left[(1-p_N)\cdot(\rho_{00}^{E}+p\cdot\rho_{01}^{E})+p_N\cdot(p\cdot\rho_{10}^{E}+{p^2}\cdot\rho_{11}^{E})\right] \\ \nonumber
    &+|1\rangle\langle1|\otimes\left[(1-p_N)\cdot(p\cdot\rho_{10}^{E}+{p^2}\cdot\rho_{11}^{E})+p_N\cdot(\rho_{00}^{E}+p\cdot\rho_{01}^{E}) \right] \},
\end{align}
where $\mathcal{N}_{p_N}$ denotes the set of string after noisy preprocessing given post-selected set $\mathcal{V}_p$, and $\mathcal{V}_p\rightarrow\mathcal{N}_{p_N}$ represents the change of remaining strings. Meanwhile, the distribution $\hat{P}(a,b|\bar{x},\bar{y},\mathcal{V}_p\rightarrow\mathcal{N}_{p_N})$ between Alice and Bob becomes
\begin{align}
    \hat{P}(a,b|\bar{x},\bar{y},\mathcal{V}_p\rightarrow\mathcal{N}_{p_N})=(1-p_{N})\widetilde{P}(a,b|\bar{x},\bar{y},\mathcal{V}_p)+p_{N}\widetilde{P}(a\oplus1,b|\bar{x},\bar{y},\mathcal{V}_p).
\end{align}
The cost of one-way error correction from Alice to Bob, $H(\hat{A}_{\bar{x}}|B_{\bar{y}},\mathcal{V}_p\rightarrow\mathcal{N}_{p_N})$, is given by
\begin{equation}
    H(\hat{A}_{\bar{x}}|B_{\bar{y}},\mathcal{V}_p\rightarrow\mathcal{N}_{p_N})= \sum_{a,b}{h[\hat{P}(a,b|\bar{x},\bar{y},\mathcal{V}_p\rightarrow\mathcal{N}_{p_N})]}-\sum_{b}h[\hat{P}(0,b|\bar{x},\bar{y},\mathcal{V}_p\rightarrow\mathcal{N}_{p_N})+\hat{P}(1,b|\bar{x},\bar{y},\mathcal{V}_p\rightarrow\mathcal{N}_{p_N})],
\end{equation}
where $h(x)$ is defined as $h(x)=-x\log_2{x}$. The notation $\hat{A}$ represents Alice's string after random post-selection and noisy preprocessing.

Given the cost of error correction $H(\hat{A}_{\bar{x}}|B_{\bar{y}},\mathcal{V}_p\rightarrow\mathcal{N}_{p_N})$, we employ a recent developed method~\cite{computeDI2021} to efficiently estimate the lower bounds on the quantum conditional entropy $H(\hat{A}_{\bar{x}}|E,\mathcal{V}_p\rightarrow\mathcal{N}_{p_N})$, which quantifies the correlations between Alice and Eve. The details are shown in sec.~\ref{sec:lb}. Combined with the key rate formula (Eq.~1 in the main text), the amount of extractable secure key bits is available.

\subsection{Lower bounds on the conditional von Neumann entropy}\label{sec:lb}
With the probability distribution $\{P(a,b|x,y)\}$ observed, $H(\hat{A}_{\bar{x}}|E,\mathcal{V}_p\rightarrow\mathcal{N}_{p_N})$ can normally be estimated by finding out its lower bound. Unlike the methods based on specific Bell inequalities~\cite{nieto2014using,de2016randomness,RogerPRA2019} associated with certain outcomes, the method involving full statistics may be more efficient at raising its lower bound. Given a bipartite state $\rho_{AB}$, the conditional von Neumann entropy of $A$ given $B$ can be defined in the form of quantum relative entropy
\begin{equation}
    H(A|B)_{\rho}:=-D(\rho_{AB}\|I_{A}\otimes \rho_{B}),
\end{equation}
where $I_{A}$ is the identity matrix. We focus on the technique of converging upper bounds on the quantum relative entropy developed in Ref.~\cite{computeDI2021}, please refer to the original text for other detailed proofs and analyses.

Referring to Ref.~\cite{computeDI2021}, the related derivations and proofs below are developed in the language of von Neumann algebra. For a linear functional $\rho$ defined on a von Neumann algebra $\mathcal{A}$, $\rho:\mathcal{A} \rightarrow \mathbb{C}$, $\rho$ is positive if $\rho(a^{*}a)\ge0$ for all $a \in \mathcal{A}$. Particularly, a positive linear functional $\rho$ is said to be a state if $\rho(I)=1$. Besides, a positive $\rho$ can also be defined by a trace-class operator $\Tilde{\rho}$ by $\rho(a)=\Tr[\Tilde{\rho}a]$ for all $a \in B(H)$, where $H$ are some separable Hilbert space. To simplify notation, the same symbol $\rho$ is used for both the positive linear functional $\rho$ and the trace-class positive operator $\Tilde{\rho}$.

A sequence of converging variational upper bounds on the relative entropy is given by the following lemma:
\begin{lemma}[Brown, Fawzi, Fawzi, 2021~\cite{computeDI2021}]
\label{theorem}
let $\rho,\sigma$ be two positive linear functionals on a von Neumann algebra $\mathcal{A}$ such that $\Tr[\rho^{2}\sigma^{-1}] < \infty$. Then for any $m\in \mathbb{N}$ and the choice of $t_1,...,t_m \in (0,1]$ and $w_1,...,w_m >0$, we have
\begin{equation}
    D(\rho||\sigma)\le -\sum_{i=1}^{m}\frac{w_i}{t_i\ln{2}} \inf_{a\in\mathcal{A}}\{\rho(I)+\rho(a+a^{*})+(1-t_i)\rho(a^{*}a)+t_{i}\sigma(aa^{*})\}.
\end{equation}
Moreover, the right hand side converges to $D(\rho||\sigma)$ as $m\rightarrow \infty$.

In the special case where $\mathcal{A}=B(H)$ for a separable Hilbert space $H$ and $\rho$ and $\sigma$ are defined via trace-class operators on $H$ (also denoted by $\rho$ and $\sigma$), satisfying $\rho\le \lambda \sigma$ for some $\lambda \in \mathbb{R}_{+}$, we can give an explicit bound on the norm of the operators appearing in the optimization:
\begin{equation}
    D(\rho||\sigma)\le c_m - \sum_{i=1}^{m-1}\frac{w_i}{t_i\ln{2}} \inf_{Z\in B(H),||Z||\le \alpha_{i}} \{\Tr[\rho(Z+Z^{*})] +(1-t_{i})\Tr[\rho Z^{*}Z]+t_i\Tr[\sigma ZZ^{*}]\},
\end{equation}
where $c_m=\frac{1}{m^2}\frac{\lambda \Tr[\rho]}{\ln{2}}-\sum_{i=1}^{m}\frac{w_i\Tr[\rho] }{t_i\ln{2}}$ and $\alpha_{i}=\frac{3}{2}\max\{\frac{1}{t_i},\frac{\lambda}{1-t_i}\}$. Moreover, the right hand side converges to $D(\rho||\sigma)$ as $m\rightarrow \infty$.
\end{lemma}

With Lemma~1, we can obtain a sequence of lower bounds on the conditional von Neumann entropy using the NPA hierarchy. In this work, $\rho$ and $\sigma$ are defined by trace-class operators instead of general positive linear functionals on $\mathcal{A}$.

We shall now estimate the lower bounds on $H(\hat{A}_{\bar{x}}|E,\mathcal{V}_p\rightarrow\mathcal{N}_{p_N})$ by a converging sequence of optimizations that can be subsequently computed using the NPA hierarchy. Given the device-independent protocol with random post-selection and noisy preprocessing, the following proof can be regarded as an application based on the original text~\cite{computeDI2021}.

\begin{proposition}
\label{proposition}
Let $m \in \mathbb{N}$ and Let $t_1,...,t_m$ and $w_1,...,w_m$ be the nodes and weights of an m-point Gauss-Radau quadrature on $[0,1]$ with  an endpoint $t_m=1$. Let $\rho_{ABE}$ be the initial quantum state shared between the devices of Alice, Bob and Eve and let $\{M^{A}_{a|\bar{x}}\}$ and $\{M^{B}_{{b}|\bar{y}}\}$ denote the measurements operators performed by Alice's and Bob's device in response to the input $X=\bar{x}$, $Y=\bar{y}$ in ``key-generation round", respectively. Let $M_a=\frac{1}{p_{\mathcal{V}_p}}\sum_{b=0}^{1}\omega_{ab}M^{A}_{a|\bar{x}}\otimes M^{B}_{b|\bar{y}}$ and $\hat{M}_{a}=(1-p_N)\cdot M_{a} + p_N\cdot M_{a\oplus1}$ for $a\in\{0,1\}$. Furthermore for $i=1,...,m-1$, let $\alpha_{i}=\frac{3}{2}\max\{\frac{1}{t_i},\frac{1}{1-t_i}\}$. Then $H(\hat{A}_{\bar{x}}|E,\mathcal{V}_p\rightarrow\mathcal{N}_{p_N})$ is never smaller than
\begin{align}
     c_m+\sum_{i=1}^{m-1}\frac{w_i}{t_i\ln{2}}\sum_{a}\inf_{Z_{a}\in B(E)}& \Tr{\left[\rho_{AE}\left( \hat{M}_{a}\otimes(Z_{a}+Z_{a}^{*}+(1-t_i)Z_{a}^{*}Z_{a})+t_i(\hat{M}_0+\hat{M}_1)\otimes Z_{a}Z_{a}^{*} \right)\right]} \\ \nonumber
     s.t.& \quad ||Z_{a}||\le \alpha_{i},
\end{align}
where $c_m=-\frac{1}{m^2\ln{2}}+\sum_{i=1}^{m}\frac{w_i}{t_i\ln{2}}$. Moreover the lower bounds converge to $H(\hat{A}_{\bar{x}}|E,\mathcal{V}_p\rightarrow\mathcal{N}_{p_N})$ as $m\rightarrow\infty$.
\end{proposition}

\begin{proof}
Let $\rho_{{A}E|\mathcal{V}_p\rightarrow\mathcal{N}_{p_N}}=\sum_{a}|a\rangle\langle a|\otimes \rho_{E|\mathcal{V}_p\rightarrow\mathcal{N}_{p_N}}(a,\bar{x})$ be the final classical-quantum state after Alice and Bob have performed the data preprocessing to the input $(\bar{x},\bar{y})$, i.e., $\rho_{E|\mathcal{V}_p\rightarrow\mathcal{N}_{p_N}}(a,\bar{x})=\Tr_{AB}[\rho_{ABE}( \hat{M}_{a} \otimes I_{E})]$. Then for any $m\in \mathbb{N}$ we have
\begin{align}
H(\hat{A}_{\bar{x}}|E,\mathcal{V}_p\rightarrow\mathcal{N}_{p_N})&=-D(\rho_{ABE|\mathcal{V}_p\rightarrow\mathcal{N}_{p_N}}||I_{A}\otimes\rho_{E|\mathcal{V}_p\rightarrow\mathcal{N}_{p_N}}) \\ \nonumber
&\ge c_m + \sum_{i=1}^{m-1}\frac{w_i}{t_i\ln{2}} \inf_{Z\in B(H)} \{\Tr[\rho_{{A}E|\mathcal{V}_p\rightarrow\mathcal{N}_{p_N}}(Z+Z^{*})] \\ \nonumber
&\quad +(1-t_{i})\Tr[\rho_{{A}E|\mathcal{V}_p\rightarrow\mathcal{N}_{p_N}} Z^{*}Z]+t_i\Tr[(I_{A}\otimes\rho_{E|\mathcal{V}_p\rightarrow\mathcal{N}_{p_N}}) ZZ^{*}]\} \\ \nonumber
&  s.t. \quad ||Z||\le \frac{3}{2}\max\{\frac{1}{t_i},\frac{\lambda}{1-t_i}\}
\end{align}
where Lemma~1 has been used with $c_m=-\frac{\lambda}{m^2\ln{2}}+\sum_{i=1}^{m}\frac{w_i}{t_i\ln{2}}$ and $\lambda$ is some real number such $\rho_{{A}E|\mathcal{V},\mathcal{N}}\le \lambda I_{A}\otimes\rho_{E|\mathcal{V},\mathcal{N}}$. As the classical system $A$ is in a finite dimension, i.e., $I_{A}=\sum_{a=0,1}|a\rangle\langle a|$, we can write the operator $Z=\sum_{ab}|a\rangle\langle b|\otimes Z_{(a,b)}$ with operators $Z_{(a,b)}\in B(E)$. Then for the first term we have
\begin{align}
    \Tr[\rho_{{A}E|\mathcal{V}_p\rightarrow\mathcal{N}_{p_N}}(Z+Z^{*})]&=\sum_{a}\Tr[\rho_{E|\mathcal{V}_p\rightarrow\mathcal{N}_{p_N}}{(a,\bar{x})}(Z_{(a,a)}+Z_{(a,a)}^{*})] \\ \nonumber
    &=\sum_{a}\Tr[\Tr_{AB}[\rho_{ABE}(\hat{M}_a\otimes I_{E})](Z_{(a,a)}+Z_{(a,a)}^{*})] \\ \nonumber
    &=\sum_{a}\Tr[\rho_{ABE}\hat{M}_a\otimes(Z_{(a,a)}+Z_{(a,a)}^{*})].
\end{align}
Repeating this for the second term, we have
\begin{align}
    \Tr[\rho_{{A}E|\mathcal{V}_p\rightarrow\mathcal{N}_{p_N}}(Z^{*}Z)]&=\sum_{ab}\Tr[\rho_{E|\mathcal{V}_p\rightarrow\mathcal{N}_{p_N}}{(a,\bar{x})}(Z_{(a,b)}^{*}Z_{(a,b)})] \\ \nonumber
    &\ge\sum_{a}\Tr[\rho_{E|\mathcal{V}_p\rightarrow\mathcal{N}_{p_N}}{(a,\bar{x})}(Z_{(a,a)}^{*}Z_{(a,a)})] \\ \nonumber
    &=\sum_{a}\Tr[\rho_{ABE}\hat{M}_a\otimes Z_{(a,a)}^{*}Z_{(a,a)}].
\end{align}
where on the second line we have used the fact $\sum_{b}Z_{(a,b)}^{*}Z_{(a,b)}\ge Z_{(a,a)}^{*}Z_{(a,a)}$. Finally, for the third term
\begin{align}
    \Tr[(I_{{A}}\otimes\rho_{E|\mathcal{V}_p\rightarrow\mathcal{N}_{p_N}})ZZ^{*}]&=\sum_{ab}\Tr[\rho_{E|\mathcal{V}_p\rightarrow\mathcal{N}_{p_N}}Z_{(a,b)}Z_{(a,b)}^{*}] \\ \nonumber
    &\ge \sum_{a}\Tr[\rho_{E|\mathcal{V}_p\rightarrow\mathcal{N}_{p_N}}Z_{(a,a)}Z_{(a,a)}^{*}] \\ \nonumber
    &= \sum_{a}\Tr[\rho_{ABE}(\hat{M}_0+\hat{M}_1)\otimes Z_{(a,a)}Z_{(a,a)}^{*}]
\end{align}
As stated in the Ref.~\cite{computeDI2021}, one only need consider the operator $Z$ that are block diagonal, i.e., $Z=\sum_{a}|a\rangle\langle a|\otimes Z_{a}$. We hence recover the objective function stated in the proposition. Since $\rho_{{A}E|\mathcal{V}_p\rightarrow\mathcal{N}_{p_N}}$ is a classical-quantum state we have $\rho_{{A}E|\mathcal{V}_p\rightarrow\mathcal{N}_{p_N}}\le I_{A}\otimes\rho_{E|\mathcal{V}_p\rightarrow\mathcal{N}_{p_N}}$ and we can set $\lambda=1$.
\end{proof}

Proposition~1 provides a converging sequence of lower bounds on the conditional von Neumann entropy. The lower bounds on the rate of our device-independent QKD protocol can be obtained by including the optimizations of all states, measurements and Hilbert spaces, subjecting to any constraints on the joint probability distribution $\{P(a,b|x,y)\}$ which involves all outcomes without post-selection to close detection loopholes. Referring to Ref.~\cite{computeDI2021}, a lower bound on $H(\hat{A}_{\bar{x}}|E,\mathcal{V}_p\rightarrow\mathcal{N}_{p_N})$ is given by the following optimization problem:
\begin{align}
     c_m+\sum_{i=1}^{m-1}\frac{w_i}{t_i\ln{2}}\inf & \sum_{a=0}^{1} \left[\langle\psi| \hat{M}_{a}(Z_{a,i}+Z_{a,i}^{*}+(1-t_i)Z_{a,i}^{*}Z_{a,i})+t_i(\hat{M}_0+\hat{M}_1) Z_{a,i}Z_{a,i}^{*}|\psi\rangle\right] \\ \nonumber
     s.t.& \quad \langle\psi|M^A_{a|x} M^B_{b|y}|\psi\rangle = P(ab|xy), \quad \text{for all $a, b, x, y$} \\ \nonumber
     & \quad \sum_{a}M^A_{a|x}=\sum_{b}M^B_{b|y}=I,\quad \text{for all $x, y$} \\ \nonumber
     & \quad M^A_{a|x}\ge0, \quad M^B_{b|y}\ge0 \quad \text{for all $a, b, x, y$} \\ \nonumber
     & \quad Z_{a,i}Z_{a,i}^{*}\le \alpha_{i} \quad \text{for all $a,i$} \\ \nonumber
     & \quad Z_{a,i}^{*}Z_{a,i}\le \alpha_{i} \quad \text{for all $a,i$} \\ \nonumber
     & \quad [M^A_{a|x},M^B_{b|y}]=[M^A_{a|x},Z_{a,i}^{(*)}]=[M^B_{a|x},Z_{a,i}^{(*)}]=0 \quad \text{for all $a, b, x, y, i$} \\ \nonumber
     & \quad M^A_{a|x}, M^B_{b|y}, Z_{a,i} \in B(H) \quad \text{for all $a, b, x, y, i$}
\label{SDP}
 \end{align}
where we have considered the pure states and dropped the tensor product structure. Using the Navascu\'{e}s-Pironio-Ac\'{\i}n (NPA) hierarchy~\cite{NPA,NPA2}, we can relax the optimization to a sequence of semidefinite programs (SDPs).

\subsection{Simulations}\label{sec:sim}
In order to choose suitable experimental parameters, we perform numerical calculations and simulations based on a model that includes non-negligible defects in practical scenarios, such as the dark count of detectors, the non-ideal fidelity of entanglement state, the multiphoton effects due to the photon-number distribution of the entanglement source. These imperfections turns out to be the main cause of additional unwanted noise that suppresses Bell violations to low levels and makes the experimental demonstration of photonic device-independent QKD nontrivial~\cite{DI2018}.

In our experiment, the spontaneous parametric down‐conversion (SPDC) source is employed to produce bipartite quantum states. A polarized beam splitter (PBS) and a superconducting nanowire single photon detector (SNSPD) are employed to perform measurements on each side, where a set of wave-plates are used to choose the local measurement settings. There are two possible outcomes for each measurement: click or non-click on the detector. We denote a click event as bit $``0"$ and a non-click event as bit $``1"$.

Ideally, the prepared quantum state is in the form of
\begin{equation}
 \ket\psi=\frac{1}{\sqrt{1+r^2}}(\ket{HV}+r\ket{VH}),
\label{eq:phi}
\end{equation}
where $r\in[0,1]$. With the imperfections in an experimental scenario, the prepared state can be given in the form of Werner state~\cite{werner1989}
\begin{equation}
 \rho=V\times \frac{1}{1+r^2}
\begin{pmatrix}
 0 & 0 & 0   & 0 \\
 0 & 1 & r   & 0 \\
 0 & r & r^2 & 0 \\
 0 & 0 & 0   & 0
\end{pmatrix}
+(1-V)\times \frac{I}{4},
\end{equation}
where $V$ denotes the visibility along a certain measurement direction. We set $V$ according to the fidelity $F$ result in section~\ref{sec:region} and the relation between fidelity $F$ and $v$, which is $F = (1+3V)/4$.

Considering the polarization-entangled state generated from SPDC source (see equation~\ref{eq:phi}), the projective measurements are restricted on the $x$--$z$ plane of the Bloch-sphere and are in the form of $\Pi(\phi)=\cos(\phi)\sigma_z+\sin(\phi)\sigma_x$ for $\phi \in [-\pi,\pi]$. The measurement results of each pair of entangled photons can be written as $\Tilde{p}(a,b|x,y)=\Tr[\rho (M^A_{a|x}\otimes M^B_{b|y})]$, where $M^A_{a|x}$ and $M^B_{b|y}$ are defined as
\begin{align}
    M^A_{0|x}&=[1-(1-p_d)(1-\eta_A)]\frac{I+\Pi(\phi_{x})}{2}+p_d\cdot \frac{I-\Pi(\phi_{x})}{2}, M^A_{1|x}=I-M^A_{0|x}\\ \nonumber
    M^B_{0|y}&=[1-(1-p_d)(1-\eta_B)]\frac{I+\Pi(\phi_{y})}{2}+p_d\cdot \frac{I-\Pi(\phi_{y})}{2}, M^B_{1|x}=I-M^B_{0|x}, \\ \nonumber
\end{align}
where $\eta_A, \eta_B \in [0,1]$ are the detection efficiencies of entangled photons from source to Alice and Bob, respectively. $p_d = 10^{-6}$ is the dark count probability in each measurement.

The photon pair distribution of a SPDC source follows a Poisson distribution~\cite{CCLim_SPDC,Kiyohara:16,liu2018high}, i.e.
\begin{equation}
    P(n)=e^{-\mu}\frac{{\mu}^n}{n!},
\end{equation}
where $\mu$ is the mean photon number. The source may emit vacuum with probability $P(0)$, one pair with probability $P(1)$, two pairs with probability $P(2)$, and so on. Considering the rapid decrease in probability as the number of photon pairs increases, we only involve the effects of three and below multiphoton pairs. In each round, four possible and also mutually exclusive outcomes make up the set $\mathcal{Z}: ab=\{00,01,10,11\}$, where the $k$-th outcome is denoted as $\mathcal{Z}_k$. Assuming a specific input setting $(x,y)$, the probabilities corresponding to the four outcomes with a single pair of photons are denoted respectively as $p_{k}(x,y)$. Then the probabilities with two-pair photons can be given by
\begin{equation}
    p_{k}^{\text{2-pairs}}(x,y)=\sum_{i,j}\beta_{i,j}^{(k)}p_{i}(x,y)p_{j}(x,y), \text{for $k$=1,2,3,4.}
\label{eq:2pair}
\end{equation}
where the four coefficients $\beta_{i,j}^{k}$ have the following form:
\begin{equation}
 \beta_{i,j}^{(1)}=
\begin{pmatrix}
 1 & 1 & 1   & 1 \\
 1 & 0 & 1   & 0 \\
 1 & 1 & 0   & 0 \\
 1 & 0 & 0   & 0
\end{pmatrix},
\beta_{i,j}^{(2)}=
\begin{pmatrix}
 0 & 0 & 0   & 0 \\
 0 & 1 & 0   & 1 \\
 0 & 0 & 0   & 0 \\
 0 & 1 & 0   & 0
\end{pmatrix},
\beta_{i,j}^{(3)}=
\begin{pmatrix}
 0 & 0 & 0   & 0 \\
 0 & 0 & 0   & 0 \\
 0 & 0 & 1   & 1 \\
 0 & 0 & 1   & 0
\end{pmatrix},
\beta_{i,j}^{(4)}=
\begin{pmatrix}
 0 & 0 & 0   & 0 \\
 0 & 0 & 0   & 0 \\
 0 & 0 & 0   & 0 \\
 0 & 0 & 0   & 1
\end{pmatrix}.
\end{equation}
The probabilities of three-pair photons can be derived according to Eq.~\ref{eq:2pair},
\begin{equation}
    p_{k}^{\text{3-pairs}}(x,y)=\sum_{i,j}\beta_{i,j}^{k}p_{i}^{\text{2-pairs}}(x,y)p_{j}(x,y).
\end{equation}
Combined these results, the joint conditional distribution $P(a,b|x,y)$ can be calculated as
\begin{align}
    P(a,b|x,y,ab=\mathcal{Z}_k)&=\exp(-u)\sum_{n=0}^{3}\frac{u^n}{n!}\cdot p_{k}^{\text{n-pairs}}(x,y), \\ \nonumber
\end{align}
where $u$ denotes the mean photon number per pulse. It should be noted that, since the dark count is independent of the input measurement settings, the probabilities of no photons are $\{p_{k}^{\text{0-pairs}}(x,y), \forall k\} = \{p_d^{2},p_d(1-p_d),p_d(1-p_d),(1-p_d)^2\}$.

In general, we develop a numerical model for optimizations based on actual defects, which performed well in our previous work~\cite{liu2018high, li2018test,liu2018device,li2021experimental,liu2021device}. The main results are shown in Fig.~\ref{fig3}, where the purple solid curve represents the results according to the performance of our system. We compute the key rate with the Gauss-Radau quadrature~\cite{davis2007methods,computeDI2021} of $m = 8$ nodes and the results are calculated at a relaxation level $2 + ABZ + AZZ$.

Assuming that no preprocessing methods is involved, we consider the variation of the secure key rate with the detection efficiency in the ideal and non-ideal cases separately. The results are represented by an orange dashed curve and a blue dotted-dashed curve in the Fig.~\ref{fig3}, respectively. It shows the impact of considering imperfections when calculating the key rate. In order to guarantee a positive key rate, ideally only the efficiency of over $83\%$ is required, but considering the actual imperfect quantum state, it must exceed the efficiency of $91.2\%$.

When random post-selection and noisy preprocessing are applied, the efficiency threshold for achieving a positive key can be effectively reduced to $86.2\%$, which is still unattainable for previous photonic Bell-test type experiments ~\cite{giustina2013bell,christensen2013detection,shalm2015strong,giustina2015significant,li2018test,liu2018high,shen2018randomness,bierhorst2018experimentally,liu2018device,zhang2020experimental,liu2021device,li2021experimental,shalm2021device}.
In the process of calculating the key rate, we adopt a genetic algorithm and traverse the following parameters, $p_N$, $p$, $r$, $u$ and measurement settings $(\phi_{x_1},\phi_{x_2},\phi_{y_1},\phi_{y_2},\phi_{y_3})$, to find the optimized key rate and the corresponding experimental parameters.

\begin{figure}[htbp]
\centering
\includegraphics[width=0.6\linewidth]{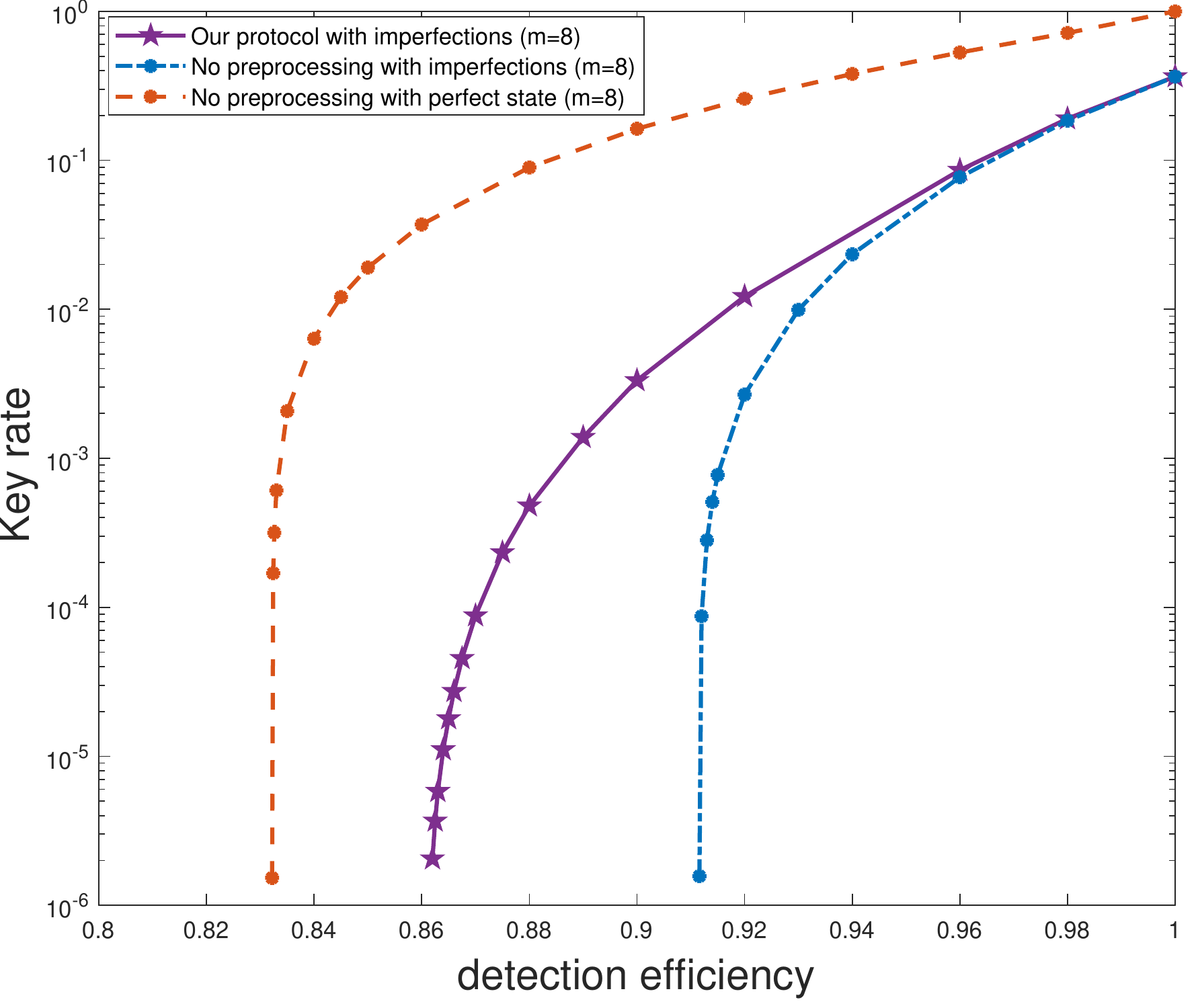}
\caption{{\bf Variation of key rate with different detection efficiencies.} Without any random post-selection and noisy preprocessing, the threshold efficiency for key generation is about $83\%$ with an ideal pure state (see orange dot-dashed curve), but increases to about $91.2\%$ considering practical imperfections (see blue dot-dashed curve). In our device-independent protocol, the threshold of efficiency is about $86.2\%$ (see purple solid curve).}
\label{fig3}
\end{figure}


In practice, to actually produce a key, more information-processing processes, such as error correction, authentication and privacy amplification, remains to be done. In table~\ref{tab:eff2}, we present the detection efficiency threshold when accounting for the typical efficiencies of error correction algorithms. To our knowledge, the best error correction inefficiency is $f_{e}=1.06$~\cite{Tang2021fe}. With this error correction algorithm, we find that it will request an efficiency threshold of $93.2\%$ to achieve a positive key. Meanwhile, since our work is a proof-of-principle verification, we choose the Shannon limit $f_e = 1.0$, which is reachable in the case of infinite data size.

\begin{table}[htbp]
\caption{Detection efficiency threshold accounting for the efficiencies of error correction algorithms.}
\begin{tabular}{c|cccc}
\hline \hline
Error correction efficiency $f_{e}$  & 1.06    & 1.01    & 1.001   & 1.0001  \\ \hline
Detection efficiency threshold         & 93.2\%  & 90.5\%  & 88.1\%  & 86.8\%  \\
\hline \hline
\end{tabular}
\label{tab:eff2}
\end{table}

\subsection{Underlying quantum distribution estimation}~\label{sec:norm}
The quantum conditional von Neumann entropy $H(\hat{A}_{\bar{x}}|E,\mathcal{V}_p\rightarrow\mathcal{N}_{p_N})$ is subject to the joint probability distributions $\{P(a,b|x,y)\}$ constraints. Given a bipartite quantum state $\rho$ shared between Alice and Bob, and the local measurement operators $A_{a|x}$ and $B_{b|y}$, where $a$($b$) and $x$($y$) represent the outcome and the measurements settings for Alice (Bob), then the joint probability distributions $\{f(a,b|x,y)\}$ can be obtained by the corresponding measurements $f(a,b|x,y)=\Tr(\rho A_{a|x}\otimes B_{b|y})$. Moreover, quantum distributions theoretically satisfy the no-signaling conditions~\cite{popescu1994quantum}, i.e., there is no information exchange between the preparation of inputs and the preparation of entangled photon pairs, or between the measurement setting on one side and outcome on the other, which has been included in the device-independent QKD regime. All these $f(a,b|x,y)$ satisfying no-signaling conditions compose a set $\mathcal{Q}$.

In the practical scenarios, the observed probability distribution $P(a,b|x,y)$ can be obtained by counting the frequency of different outcomes under different inputs, i.e., $P(a,b|x,y)=\frac{N_{a,b,x,y}}{N_{x,y}}, \forall a,b,x,y$, where $N_{a,b,x,y}$ is the counting of events with specific inputs and outcomes $(a,b,x,y)$, and $N_{x,y}=\sum_{a,b}N_{a,b,x,y}$ is the total number of rounds corresponding to the measurement choice $(x,y)$. However, due to the statistical fluctuations, $P(a,b|x,y)$ is likely not to satisfy the no-signaling constraints, so it cannot be directly used to compute the key rate.

Therefore, we adopt the least-square-error estimation technique~\cite{LSesti} to determine a probability distribution $\mathcal{P}(a,b|x,y)$ satisfying the no-signaling conditions and Tsirelson’s bounds~\cite{barrett2005nonlocal}.
\begin{align}
   &\min ||\vec{P}-\vec{\mathcal{P}}||_{2} \\ \nonumber
   &\text{s.t} \quad  \mathcal{P}(a,b|x,y) \in \mathcal{Q},
\end{align}
where vector $\vec{P}$ denotes the observed probability distributions $\{P(a,b|x,y)\}$ and $||.||_{p}$ denotes the $p$-norm.

\subsection{Dual solution and Confidence region}~\label{sec:region}
The dual solution of proposition~\ref{proposition} provides us with an affine function
\begin{equation}
\label{dual0}
    g(\vec{P})=\alpha+\sum_{a,b,x,y}\lambda_{a,b,x,y}P(a,b|x,y),
\end{equation}
which is always an upper bound on the primal program~\ref{proposition}. As developed in Ref.~\cite{nieto2014using}, $g(\vec{P})$ can be seen as the optimal Bell expressions, i.e., the randomness contained in $H(\hat{A}_{\bar{x}}|E,\mathcal{V}_p\rightarrow\mathcal{N}_{p_N})$ certified from the Bell expressions $g(\vec{P})$ is equal to that can be certified from the full set of probabilities $\{P(a,b|x,y)\}$.

In our protocol, the observed probability distributions $P(a,b|x,y)$ are composed of 24 different probabilities according to different inputs $(x,y)$ and outcomes $(a,b)$. According to the Tsirelson’s bounds~\cite{barrett2005nonlocal}, the quantum correlations can be expressed with 24 Bell expressions denoted as $e_{t} t=1,2...24$, where each $t$ corresponds to one possible event of $(a,b,x,y)$. Combined with the constraints of normalization and no-signalling conditions, $P(a,b|x,y)$ can be uniquely determined by only 11 of them, $h_{k} k=1,2...11$, where $h_1$ and $h_2$ are the Alice’s marginal probabilities $P_{A}(1|x), x=0,1$, and $h_3$, $h_4$ and $h_5$ are the Bob’s marginal probabilities $P_{B}(1|y), y=1,2,3$, and $h_{6},...,h_{11}$ are bipartite correlations $P(1,1|x,y), x=0,1 and y=1,2,3$. In this sense, the dual solution~\ref{dual0} can be rewritten as
\begin{equation}
\label{dual1}
    g(\vec{P})=\beta + \sum_{t}\gamma_{t}h_{t},
\end{equation}
where $\beta$ is a constant and $\gamma_{t}$ is the corresponding coefficient for $h_{t}$. Considering the dual solution of the simulation with $87.5\%$ efficiency as an example, Eq.~\ref{dual1} is in the following form:
\begin{align}
    g(\hat{P})=&-1.2951+0.8564P(1,1|1,1)+0.8437P(1,1|1,2)-0.05783P(1,1|1,3)+0.8565P(1,1|2,1) \\
    &-1.1736P(1,1|2,2)+0.003611P(1,1|2,3)+0.09415P_{A}(1|1)+0.3155P_{A}(1|2)  \\
    &-0.8432P_B(1|1)+0.3312P_{B}(1|2)+0.9027P_{B}(1|3).
\end{align}

Using the straightforward estimator~\cite{Nieto_Silleras_2018} for the Bell expectations $g(\vec{P})$, we can construct corresponding confidence regions for $N-$round experiment. Let $\mathcal{X}(a,b,x,y)$ be the indicator function for the event $\{a,b,x,y\}$, i.e., $\mathcal{X}$ = 1 if the event $\{a,b,x,y\}$ is observed and
$\mathcal{X}=0$ otherwise. Given the following random variable for the $i_{\text{th}}$-round
\begin{equation}
    g'_{i}=g_{i}-\alpha=\sum_{a,b,x,y}\lambda_{a,b,x,y}\frac{\mathcal{X}(a=a_i,b=b_i,x=x_i,y=y_i)}{P(x_i,y_i)},
\end{equation}

and the assumptions that the devices are memoryless and behave identically and independently in the implementation of the protocol, we thus could have the expected value $\text{E}{[g'_{i}]}=g(\vec{P})-\alpha$ and $|g_{i}|\le \lambda/q$, where $\lambda=\max|\lambda_{a,b,x,y}|$ and $q=\max{P(x,y)}$. The resulting estimators $\hat{g}'$ for the Bell expectations $g(\vec{P})-\alpha$ by the observed frequencies $\hat{g}'_{i}$ are given by
\begin{equation}
    \hat{g}'=\frac{1}{n}\sum_{i=1}^{n}\hat{g}'_{i}= \frac{1}{n}\sum_{i=1}^{n}\frac{\lambda_{a_i,b_i,x_i,y_i}}{P(x_i,y_i)}.
\end{equation}

Using a consequence of the Hoeffding’s inequality~\cite{pironio2010random, Nieto_Silleras_2018}, we have
\begin{equation}
    \text{Pr}[|g(\vec{P})-(\hat{g}'+\alpha)|\ge \delta]\le 2\epsilon.
\end{equation}
where $\epsilon=\exp\left(\frac{-2n\delta^2}{(2\lambda/q)^2}\right)$.

\section{System characterization}

\subsection{Determination of single photon efficiency}
In the experiment, we optimize the system efficiency by choosing appropriate pump beam waist and parametric beam waist size~\cite{NIST_SPDC}. The pump light spreads about $70~cm$ and is focused to the center of the PPKTP by an aspherical lens with $f = 8~mm$. The radius of pump waist is estimated to be $200~\mu$m. In order to efficiently collect parametric light, the corresponding optimal radius of parametric waist is estimated to be $97.5~\mu$m, and two spherical lenses with $f = 150~mm$ and $f = -150~mm$ and an aspherical lens with $f = 9.6~mm$ are placed in sequence for collection. The distance between the two spherical lenses is about $12.5~cm$, of which the lens close to the pump light is about $33.5~cm$ from the center of the PPKTP, and the other is about $70~cm$ from the aspheric lens.

We define the single photon heralding efficiency as $\eta_A=C/N_B$ and $\eta_B=C/N_A$ for Alice and Bob, in which two-photon coincidence events $C$ and single photon detection events for Alice $N_A$ and Bob $N_B$ are measured in the experiment. The heralding efficiency is listed in Tab.~\ref{tab:OptEffAB}, where $\eta^{\mathrm{sc}}$ is the efficiency of coupling entangled photons into single mode optical fiber, $\eta^{\mathrm{so}}$ is the optical efficiency due to limited transmittance of optical elements in the source, and $\eta^{\mathrm{det}}$ is the single photon detector efficiency. $\eta^{\mathrm{so}}$, $\eta^{\mathrm{det}}$ can be measured with classical light beams and NIST-traceable power meters. The optical transmittance for all involved optical elements are listed in Tab.~\ref{tab:OptEff}, with which we obtain the optical efficiency $\eta^{so}$:

\begin{equation} \label{Eq:heraldingEff}
    \begin{cases}
\eta^{\mathrm{so}}(Alice) = \eta^{\mathrm{AS}} \times {(\eta^{\mathrm{S}})}^2 \times {(\eta^{\mathrm{DM}})}^7 \times \eta^{780/1560\mathrm{HWP}} \times \eta^{780/1560\mathrm{PBS}} \times \eta^{\mathrm{PPKTP}} \times \eta^{1560\mathrm{HWP}} \times \eta^{1560\mathrm{QWP}} \times \eta^{1560\mathrm{PBS}}\\
\eta^{\mathrm{so}}(Bob) = \eta^{\mathrm{AS}} \times {(\eta^{\mathrm{S}})}^2 \times {(\eta^{\mathrm{DM}})}^8 \times \eta^{780/1560\mathrm{HWP}} \times \eta^{780/1560\mathrm{PBS}} \times \eta^{\mathrm{PPKTP}} \times \eta^{1560\mathrm{HWP}} \times \eta^{1560\mathrm{QWP}} \times \eta^{1560\mathrm{PBS}}\\
    \end{cases}
\end{equation}

\begin{table}[htb]
\caption{\bf Characterization of optical efficiencies in the experiments. }
\begin{tabular}{c|c|ccc}
\hline
\multicolumn{1}{l|}{} & heralding efficiency ($\eta$) & $\eta^{\mathrm{sc}}$ & $\eta^{\mathrm{so}}$ & $\eta^{\mathrm{det}}$ \\ \hline
Alice                 & 87.16\%                       & 95.57\%              & 93.53\%              & 97.51\%               \\
Bob                   & 87.82\%                       & 95.74\%              & 93.03\%              & 98.60\%               \\ \hline
\end{tabular}
\label{tab:OptEffAB}
\end{table}

\begin{table}[htb]
\centering
  \caption{\bf The efficiencies of optical elements.}
\begin{tabular}{c|c|c}
\hline
Optical element & Symbol & Efficiency  \\
\hline
Aspherical lens & $\eta^{\mathrm{AS}}$ & $99.27\%\pm0.03\%$ \\
Spherical lens & $\eta^{\mathrm{S}}$ & $99.6\%\pm1.0\%$ \\
Half wave plate (780nm/1560nm) & $\eta^{780/1560\mathrm{HWP}}$ & $99.93\%\pm0.02\%$ \\
Half wave plate (1560nm) & $\eta^{1560\mathrm{HWP}}$ & $99.92\%\pm0.04\%$ \\
Quarter wave plate (1560nm) & $\eta^{1560\mathrm{QWP}}$ & $99.99\%\pm0.08\%$ \\
Polarizing beam splitter (780nm/1560nm) & $\eta^{780/1560\mathrm{PBS}}$ & $99.6\%\pm0.1\%$ \\
Polarizing beam splitter (1560nm) & $\eta^{1560\mathrm{PBS}}$ & $99.6\%\pm0.2\%$ \\
Dichoric mirror & $\eta^{\mathrm{DM}}$ & $99.46\%\pm0.03\%$ \\
PPKTP & $\eta^{\mathrm{PPKTP}}$ & $99.6\%\pm0.2\%$ \\
\hline
\end{tabular}
\label{tab:OptEff}
\end{table}

\subsection{Quantum state and measurement bases}~\label{sec:tomo}
To optimize the key rate in our experiment, we aim to create a non-maximally entangled two-photon state$\cos(\alpha)\ket{HV}+\sin(\alpha)\ket{VH}$, where $\alpha = 20.0^\circ$ and set measurement bases to be $x_1=-88.22^\circ$ and $x_2=54.29^\circ$ for Alice, and $y_1=9.75^\circ$, $y_2=21.45^\circ$ and $y_3=-1.07^\circ$ for Bob, respectively. We also optimize the mean photon number to be 0.04 to maximize the key rate.


We perform quantum state tomography measurement of the non-maximally entangled state, with result shown in Fig.~\ref{Fig:Tomo}. The state fidelity is $99.52\pm0.15\%$. We attribute the imperfection to multi-photon components, imperfect optical elements, and imperfect spatial/spectral mode matching.


\begin{figure}[htbp]
\centering
\resizebox{16cm}{!}{\includegraphics{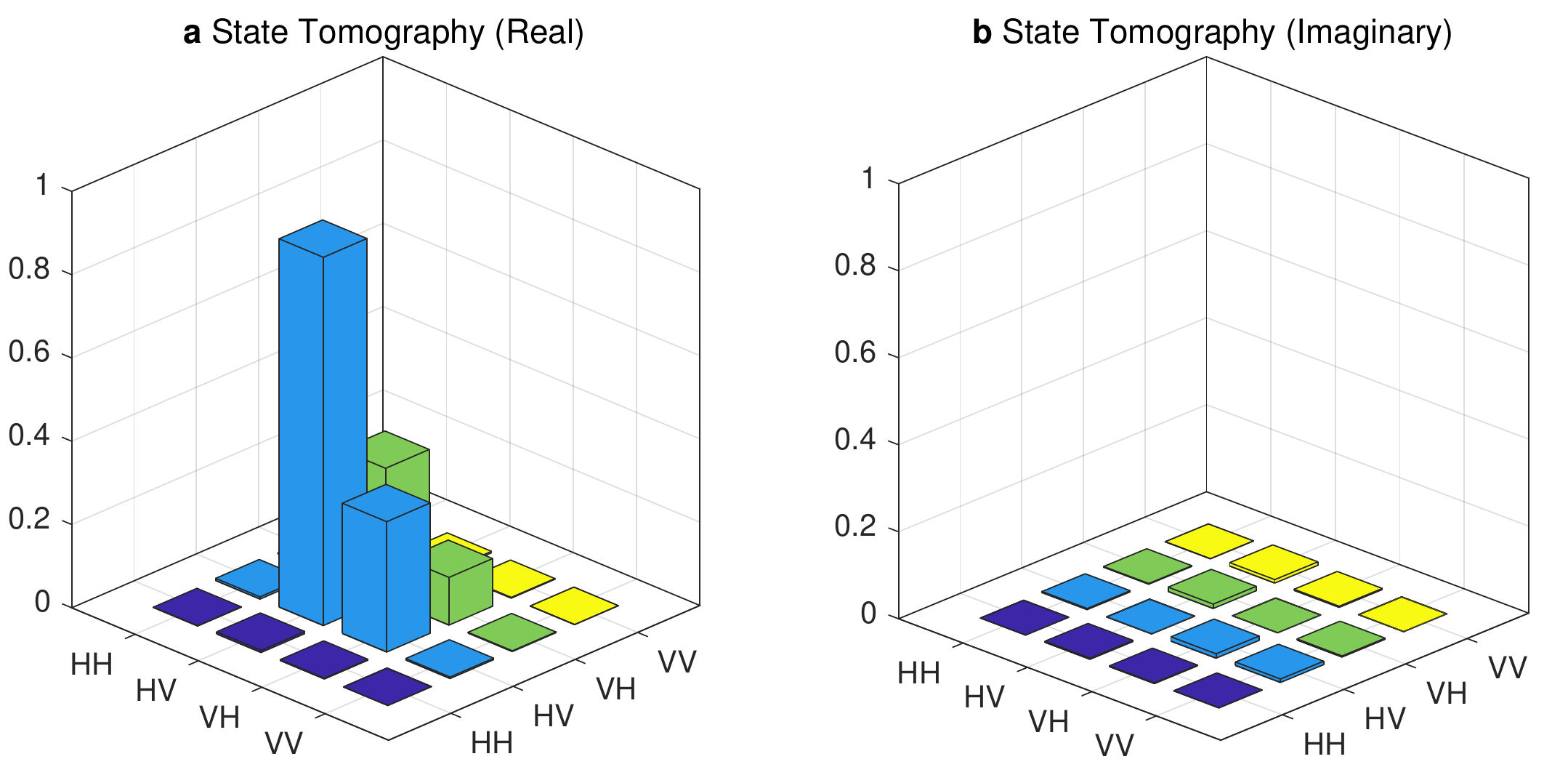}}
\caption{{\bf (color online) Tomography of the produced two-photon state in the experiments,} with real and imaginary components shown in {\bf a} and {\bf b}.}
\label{Fig:Tomo}
\end{figure}

\section{Experimental results}

\subsection{CHSH test}
Before the start of the experiment, systematic experimental calibrations are implemented to optimize a CHSH~\cite{CHSH}Bell test based on the performance of system. We optimize to create a non-maximally entangled two-photon state$\cos(\alpha)\ket{HV}+\sin(\alpha)\ket{VH}$, where $\alpha = 32.2^\circ$ and set measurement bases to be $x_1=-81.09^\circ$ and $x_2=61.46^\circ$ for Alice, and $y_1=8.18^\circ$ and $y_2=-29.37^\circ$ for Bob, respectively. We also optimize the mean photon number to be 0.62 to maximize the CHSH score. The CHSH score $\omega_{\mathrm{CHSH}}$ is given by
\begin{equation}
\omega_{\mathrm{CHSH}} = \frac{1}{2}\sum_{k,l}{\sum_{i=1}^{n_{x_iy_i=kl}}{(1+(-1)^{a_i\oplus b_i\oplus(x_i\cdot y_i)})/n_{x_iy_i=kl}}},
\end{equation}
with $(k,l)\in(0,1)\times(0,1)$ is computed to be $0.7559$. The recorded data are shown in Fig.~\ref{Fig:vio}.

\begin{figure}[htbp]
\centering
\resizebox{10cm}{!}{\includegraphics{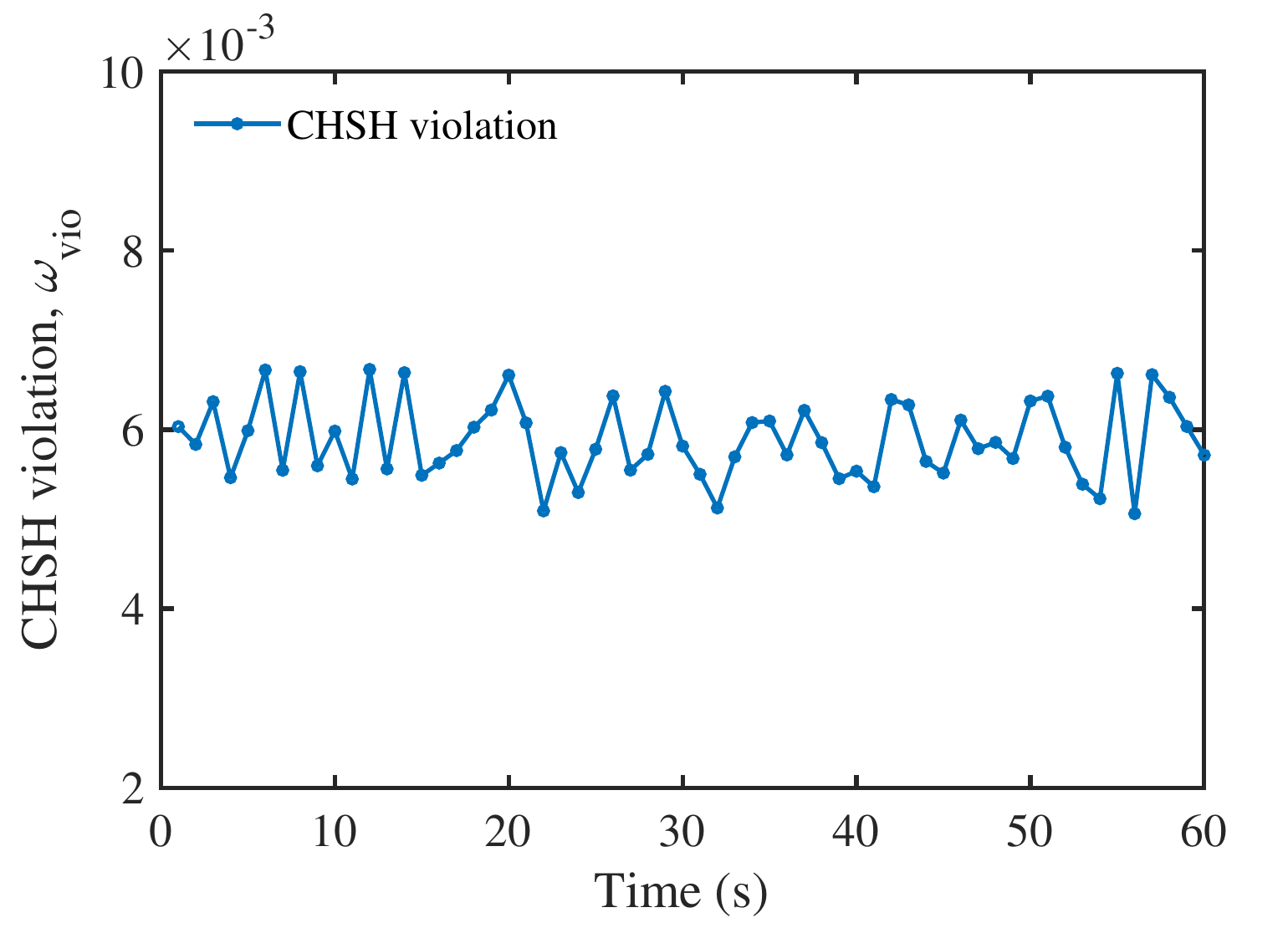}}
\caption{{\bf CHSH violation versus time.} We conduct $1.2\times10^7$ rounds of CHSH tests at a repetition rate of $200$~KHz. We choose the average value of CHSH violations $\omega_{\mathrm{vio}}$ for each $1$ second data as a point to observe its performance, which are defined as $\omega_{\mathrm{vio}}= \omega_{\mathrm{CHSH}} - 0.75$.}
\label{Fig:vio}
\end{figure}

\subsection{Experimental data analysis}
The implementation of the protocol depends on numerical studies shown in section~\ref{sec:sim}. We conduct $2.4\times10^8$ rounds of experiment for each of the six combinations of measurement settings ($x,y$) at a repetition rate of $2\times10^6$~Hz. The intensity of entangled photons, i.e. mean photon number and counts are shown in Tab.~\ref{tab:output}. According to the method described in section~\ref{sec:norm}, we first optimize the probability distribution $P(a,b|x,y)$ under the no-signaling condition~\cite{popescu1994quantum} and Tsirelson’s bounds~\cite{barrett2005nonlocal}. The results are shown in Tab.~\ref{tab:norm}. Based on the method described in section~\ref{sec:region}, we set the error probability $\epsilon=10^{-2}$. For the experiment of 20~m fiber length, $\lambda\approx3.15$ is calculated according to $P(a,b|x,y)$. Thus we could have $\delta\approx0.0015$ representing the confidence region of $H(\hat{A}_{\bar{x}}|E,\mathcal{V}_p\rightarrow\mathcal{N}_{p_N})$. While for the experiment of 80~m fiber length, $\lambda\approx2.36$, and $\delta\approx0.0011$. In the experiment of 220~m fiber length, $\lambda\approx6.66\times10^{-2}$, and $\delta\approx3.20\times10^{-5}$. It should be noted that $\delta$ is proportional to $1/\sqrt{n}$, indicating less uncertainty with more experimental rounds.

\begin{table}[htbp]
\centering
\caption{{\bf Mean photon number and counts of experimental rounds.} Recorded number of two-photon detection events for six sets of polarization state measurement bases $x=1$ or $2$ indicates two different settings and $y=1,2$ or $3$ indicates three different settings, which are obtained with the corresponding mean photon number. $a(b)=0$ or $1$ indicates that Alice(Bob) detects a photon or not.}
\begin{tabular}{c|c|c|cccc}
\hline
Fiber length/m       &\begin{tabular}[c]{@{}c@{}}Mean photon number\\ /per pulse\end{tabular} & Basis settings & $ab=11$   & $ab=10$ & $ab=01$ & $ab=00$ \\ \hline
\multirow{6}{*}{20}  & \multirow{6}{*}{0.040}                                                    & $xy = 11$     & 238565091 & 403056  & 210284  & 821569  \\
                     &                                                                          & $xy = 12$     & 237797832 & 1164958 & 225388  & 811822  \\
                     &                                                                          & $xy = 13$     & 238783108 & 162648  & 160633  & 893611  \\
                     &                                                                          & $xy = 21$     & 236368886 & 222941  & 2390061 & 1018112 \\
                     &                                                                          & $xy = 22$     & 234578742 & 1973697 & 3385904 & 61657   \\
                     &                                                                          & $xy = 23$     & 236058614 & 516659  & 2882175 & 542552  \\ \hline
\multirow{6}{*}{80}  & \multirow{6}{*}{0.035}                                                    & $xy = 11$     & 238758474 & 350531  & 182413  & 708582  \\
                     &                                                                          & $xy = 12$     & 238128577 & 998486  & 190177  & 682760  \\
                     &                                                                          & $xy = 13$     & 238971836 & 142160  & 136010  & 749994  \\
                     &                                                                          & $xy = 21$     & 236913121 & 197393  & 2020739 & 868747  \\
                     &                                                                          & $xy = 22$     & 235413651 & 1674701 & 2863370 & 48278   \\
                     &                                                                          & $xy = 23$     & 236655191 & 441918  & 2442392 & 460499  \\ \hline
\multirow{6}{*}{220} & \multirow{6}{*}{0.040}                                                    & $xy = 11$     & 238531173 & 422523  & 221542  & 824762  \\
                     &                                                                          & $xy = 12$     & 237751528 & 1191525 & 236435  & 820512  \\
                     &                                                                          & $xy = 13$     & 238808826 & 168059  & 164710  & 858405  \\
                     &                                                                          & $xy = 21$     & 236390264 & 231282  & 2373297 & 1005157 \\
                     &                                                                          & $xy = 22$     & 234718386 & 1924194 & 3299222 & 58198   \\
                     &                                                                          & $xy = 23$     & 236077278 & 517035  & 2869523 & 536164  \\ \hline
\end{tabular}
\label{tab:output}.
\end{table}

\begin{table}[htbp]
\caption{The input-output probability distribution $P(a,b|x,y)$. A 2-norm estimation is applied to the data to derive a probability distribution adapted to the model used. The probabilities below are the results before noisy preprocessing and random post-selection.}
\begin{tabular}{c|c|cccc}
\hline
Fiber length/m       & \diagbox{$(x,y)$}{$(a,b)$}                    & $11$            & $10$            & $01$            & $00$            \\ \hline
\multirow{6}{*}{20}  & $11$                                          & 0.9939848305625 & 0.0016772430375 & 0.0008783402625 & 0.0034595861375 \\
                     & $12$                                          & 0.9907551763625 & 0.0049068972375 & 0.0008862111125 & 0.0034517152875 \\
                     & $13$                                          & 0.9949539201500 & 0.0007081534400 & 0.0006388507000 & 0.0036990757000 \\
                     & $21$                                          & 0.9848486222042 & 0.0008729597292 & 0.0100145486208 & 0.0042638694458 \\
                     & $22$                                          & 0.9775156493042 & 0.0082059326292 & 0.0141257381708 & 0.0001526798958 \\
                     & $23$                                          & 0.9835762722167 & 0.0021453097167 & 0.0120164986333 & 0.0022619194333 \\ \hline
\multirow{6}{*}{80}  & $11$                                          & 0.9948356614583 & 0.0014838718733 & 0.0007367281267 & 0.0029437385417 \\
                     & $12$                                          & 0.9921373177096 & 0.0041822156221 & 0.0007705468779 & 0.0029099197904 \\
                     & $13$                                          & 0.9957109010408 & 0.0006086322908 & 0.0005504093742 & 0.0031300572942 \\
                     & $21$                                          & 0.9871206246575 & 0.0007904517325 & 0.0084517649275 & 0.0036371586825 \\
                     & $22$                                          & 0.9809551559062 & 0.0069559204838 & 0.0119527086813 & 0.0001362149287 \\
                     & $23$                                          & 0.9860772142350 & 0.0018338621550 & 0.0101840961800 & 0.0019048274300 \\ \hline
\multirow{6}{*}{220} & $11$                                          & 0.9938998979150 & 0.0017579270850 & 0.0009256770850 & 0.0034164979150 \\
                     & $12$                                          & 0.9906931322925 & 0.0049646927075 & 0.0009851406225 & 0.0033570343775 \\
                     & $13$                                          & 0.9949693281250 & 0.0006884968750 & 0.0006980406250 & 0.0036441343750 \\
                     & $21$                                          & 0.9849438305550 & 0.0009706680550 & 0.0098817444450 & 0.0042037569450 \\
                     & $22$                                          & 0.9779142690975 & 0.0080002295125 & 0.0137640038175 & 0.0003214975725 \\
                     & $23$                                          & 0.9837356045150 & 0.0021788940950 & 0.0119317642350 & 0.0021537371550 \\ \hline
\end{tabular}
\label{tab:norm}.
\end{table}


%

\end{document}